\documentclass[runningheads,english,numberwithinsect,envcountsame]{llncs}

\usepackage[T1]{fontenc}
\usepackage[english]{babel}
\usepackage[utf8]{inputenc}
\usepackage{amssymb,amsmath}
\usepackage{graphicx}
\usepackage[backgroundcolor=white,linecolor=black,textsize=scriptsize]{todonotes}
\usepackage{xspace}
\usepackage[pagewise]{lineno}
\usepackage{subfig}
\usepackage{paralist}

\usepackage{thmtools}
\usepackage{thm-restate}

\usepackage{marcel_notation}

\newcommand{\Gsp}{G_{sp}}

\newcommand{\coGsp}{\overline{G_{sp}}}

\newcommand{\st}{{st}}

\newcommand{\ExtG}{G^\star_{st}}
\newcommand{\ExtH}{H^\star_{st}}

\newcommand{\DirG}{\overrightarrow{G}_{\st}}
\newcommand{\DirE}{\overrightarrow{E}}
\newcommand{\DirV}{\overrightarrow{V}}

\newcommand{\pl}{p_{\lambda}}
\newcommand{\pu}{p_{\mu}}

\let\doendproof\endproof
\renewcommand{\endproof}{\hfill $\square$\doendproof}

\title{Inserting an Edge into a Geometric Embedding\thanks{Work
was partially supported by grant WA 654/21-1 of the German Research Foundation
(DFG).}}

\author{Marcel Radermacher\inst{1} \and Ignaz Rutter\inst{2}}

\institute{
		Department of Computer Science, Karlsruhe  Institute of Technology,
	Germany \and
	Department of Computer Science and Mathematics, University
	of Passau, Germany \\
	\email{radermacher@kit.edu}, \email{rutter@fim.uni-passau.de}}

\begin{document}

\maketitle

\begin{abstract}
	The algorithm to insert an edge $e$ in linear time into a planar graph $G$
	with a minimal number of crossings on
	$e$~\cite{DBLP:journals/algorithmica/GutwengerMW05}, is a helpful tool for
	designing heuristics that minimize edge crossings in drawings of general graphs.
	Unfortunately, some graphs do not have a geometric embedding $\Gamma$ such
	that $\Gamma+e$ has the same number of crossings as the embedding $G+e$.  This
	motivates the study of the computational complexity of the following problem:
	Given a combinatorially embedded graph $G$, compute a geometric embedding
	$\Gamma$ that has the same combinatorial embedding as $G$ and that minimizes
	the crossings of $\Gamma+e$.  We give polynomial-time algorithms for special
	cases and prove that the general problem is fixed-parameter tractable in the
	number of crossings. Moreover, we show how to approximate the number of
	crossings by a factor $(\Delta-2)$, where $\Delta$ is the maximum vertex
	degree of $G$. 
\end{abstract}

\section{Introduction}

Crossing minimization is an important task for the construction of readable
drawings.  The problem of minimizing the number of crossings in a given graph is
a well-known $\cNP$-complete problem~\cite{garey1983crossing}.  A very
successful heuristic for minimizing the number of crossings in a topological
drawing of a graph $G$ is to start with a spanning planar subgraph $H$ of $G$
and to iteratively \emph{insert} the remaining edges into a drawing of~$H$.  The
edge insertion problem for a planar graph $G$ and two vertices $s,t \in V(G)$
asks to find a drawing~$\Gamma+st$ of~$G+st$ with the minimum number of
crossings such that the induced drawing~$\Gamma$ of $G$ is planar.  The problem
comes with several variants depending on whether the drawing~$\Gamma$ can be
chosen arbitrarily or is
fixed~\cite{JGAA-160,DBLP:journals/algorithmica/GutwengerMW05}.  In the planar
topological case both problems can be solved in linear time.  More general
problems such as inserting several edges
simultaneously~\cite{chimani_et_al:LIPIcs:2016:5922} or inserting a vertex
together with all its incident edges~\cite{DBLP:conf/soda/ChimaniGMW09} have
also been studied.

All these approaches have in common that they focus on topological drawings
where edges are represented as arbitrary curves between their endpoints.  By
contrast, we focus on geometric embeddings, i.e., planar straight-line drawings,
and the corresponding rectilinear crossing number. In this scenario we are only
aware of a few heuristics that compute straight-line drawings of general
graphs~\cite{DBLP:journals/corr/abs-1201-3011,doi:10.1137/1.9781611975055.12}.
Clearly, if a geometric embedding~$\Gamma$ of the input graph $G$ is provided as
part of the input, there is no choice left; we can simply insert the
straight-line segment from $s$ to $t$ into the drawing and count the number of
crossings it produces.  If, however, only the combinatorial embedding is
specified, but one may still choose the outer face and choose the vertex
positions so that this results in a straight-line drawing with the given
combinatorial embedding, then the problem becomes interesting and non-trivial.
We call this problem \emph{geometric edge insertion}.

\subsubsection*{Contribution and Outline.}

We show several results on the complexity of geometric edge insertion with a
fixed combinatorial embedding.  Namely, we give a linear-time algorithm for the
case that the maximum degree~$\Delta$ of $G$ is at most~5
(Sec.~\ref{sec:bounded_degree}).  For the general case, we give a
$(\Delta-2)$-approximation that runs in linear time.  Moreover, we give an
efficient algorithm for testing in special cases whether there exists a way to
insert the edge $st$ so that it does not produce more crossings than when we
allow to draw it as an arbitrary curve (Sec.~\ref{sec:consistent_shortest}).
Finally, we give a randomized FPT algorithm that tests in~$O(4^k n)$ time
whether an edge can be inserted with at most $k$ crossings (Sec.~\ref{sec:fpt}).  

\section{Preliminaries}

\begin{figure}[tb]
	\centering
	\includegraphics{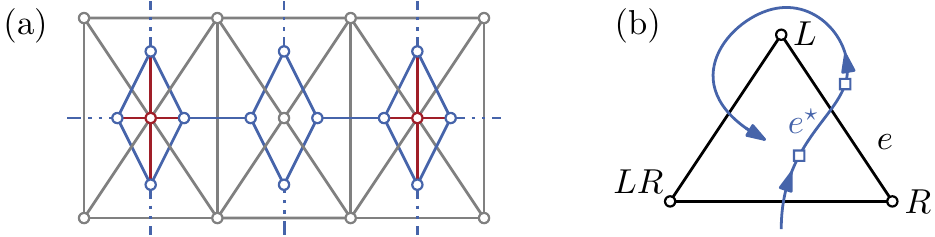}
	\caption{(a) The extended dual (red + blue) of the primal graph (grey) and the red
	vertices corresponding to $s$ and $t$. (b) Labeling induced by the blue path. 
}
	\label{fig:extended_dual}
\end{figure}

Let $G=(V, E)$ be a planar graph with a given combinatorial embedding where only
the  choice of the outer face is free. Additionally, let $s$ and $t$ be two
distinct vertices with $st \not\in E$.  Denote by $G+st$ the graph $G$ together
with the edge $st$. We want to insert the \emph{edge $st$ into the embedded
graph $G$}. That is, we seek a straight-line drawing $\Gamma$ of $G$ (with the
given embedding) such that $st$ can be inserted into $\Gamma$ with a minimum
number of crossings. In $\Gamma$, the edge $st$ starts at $s$, traverses a set
of faces and ends in $t$. Topologically, this corresponds to a path $p(\Gamma)$
from $s$ to $t$ in the \emph{extended dual $\ExtG$ of $G$}, i.e., in the dual
graph $G^\star$ plus $s$ and $t$ connected to all vertices of their
dual faces; see Fig.~\ref{fig:extended_dual}a.  The number of crossings in
$\Gamma+st$ corresponds to the length of the path minus two. However, not all
$st$-paths in $\ExtG$ are of the form $p(\Gamma)$ for a straight-line drawing
$\Gamma$ of $G$. 

\begin{figure}[tb]
	\centering
	\includegraphics{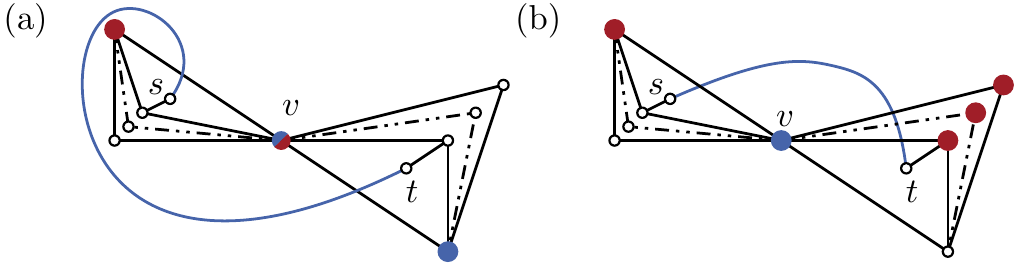}
	\caption{ Ratio between length of the shortest $st$ path and the length of a
	shortest consistent $st$-path.
	The solid black edges induce a graph of maximum degree $6$. Red vertices have
	label $L$, blue vertices have label $R$. (a) The shortest path from $s$ to $t$
	in $\ExtG$ is not consistent.
	}
	\label{fig:bad_ratio}
\end{figure}

A labeling of $G$ is a mapping $l\colon V \to \{L, R\}$ that labels
vertices as either left or right.  Consider an edge $uv$ of $G$ that
is crossed by a path $p$ such that $u$ and $v$ are to the left and to
the right of $p$, respectively.  The edge $uv$ is \emph{compatible}
with a labeling $l$ if $l(u)=L$ and $l(v) =R$.  A path $p$ of $\ExtG$
and a labeling $l$ of $G$ are \emph{compatible} if $l$ is compatible
with each edge that is crossed by $p$.  A path $p$ is
\emph{consistent} if there is a labeling of $G$ that is compatible
with $p$.  Eades et al.~\cite{10.1007/978-3-319-21840-3_25} show the
following result.

\begin{proposition}[Eades et al.~\cite{10.1007/978-3-319-21840-3_25}, Theorem~1]
	\label{prop:consistent_path}
	An $st$-path in $\ExtG$ is of the form $p(\Gamma)$ if and
	only if it is consistent, where $\Gamma$ is a geometric embedding of $G$. 
\end{proposition}

In order to minimize the number of crossings of $\Gamma+st$, we look for a
consistent $st$-path of minimum length in $\ExtG$.  Given a path $p$, it is easy
to check whether $p$ is consistent. Fig.~\ref{fig:bad_ratio} shows that the
ratio between the length of a shortest $st$-path and the length of a shortest
consistent $st$-path can be arbitrarily large. Thus, our goal is to find short
consistent $st$-paths.

Let $H=(V', E')$ be a directed acyclic graph.  A path $p = \langle v_1, v_2,
\dots, v_k\rangle$ is a \emph{directed path} if for each $1 \leq i < k$,
$v_iv_{i+1} \in E'$. It is \emph{undirected} if for each $ 1 \leq i < k$, either
$v_iv_{i+1} \in E'$ or $v_{i+1}v_i \in E'$. We refer to the number $|p|$ of edges of
a path as the \emph{length of $p$}.  Two paths $p$ and $p'$ are
\emph{edge-disjoint} if they do not share an edge. Two paths $p$ and $p'$ of an
embedded graph are \emph{non-crossing} if at each common vertex $v$, the edges
of $p$ and $p'$ incident to $v$ do not alternate in the cyclic order around
$v$ in the graph induced by $p$ and $p'$.  We denote by $p[u, v]$ the subpath of
a path $p$ from $u$ to $v$.

\section{Bounded Degree}
\label{sec:bounded_degree}

The shortest $st$-path of the graph in Fig.~\ref{fig:bad_ratio}a is not
consistent. Note that the maximum vertex degree is $6$.  In this section, we
show that every shortest $st$-path in graphs of bounded degree $3$ is
consistent, and that in each planar graph with vertex degree at most $5$,  there
is a shortest $st$-path that is consistent.  Finally, we prove that there is a
consistent $st$-path of length $(\Delta - 2) l $ in a graph with maximum vertex
degree~$\Delta$ and a shortest $st$-path of length $l$ in $\ExtG$.

Let $p$ be an $st$-path in $\ExtG$ and let $e^\star$ be an edge of $p$. An
endpoint $u$ of the primal edge $e$ of $e^\star$ is \emph{left of $e^\star$} if
it is locally left of $p$ on $e$ (Fig.~\ref{fig:extended_dual}b).  A vertex $v$
of $G$ is \emph{left} (\emph{right}) of $p$ if $v$ is left (right) of an edge of
$p$.  We now consider a labeling extended by two more labels $LR, \bot$. We
define the labeling~$l_p$ induced by $p$ as follows. Each vertex that is left
and right of $p$ gets the label $LR$. The remaining vertices that are either
left or right of $p$ get labels $L$ and $R$, respectively.  Vertices neither left nor right of
$p$ get the label $\bot$.  Obviously, there is a labeling $l$ of $G$ compatible
with $p$ if and only if $l_p$ does not use the label $LR$.

\begin{restatable}{theorem}{thmdegthree}
	\label{theorem:degree_3}
	Let $G$ be a planar embedded graph of degree at most $3$. Then every
	shortest $st$-path in $\ExtG$ is consistent.
\end{restatable}

\begin{theorem}
	\label{theorem:degree_5}
	Let $G$ be a planar embedded graph with maximum degree $5$. Then there is a
	shortest $st$-path in $\ExtG$ that is consistent.
\end{theorem}

\begin{proof}
	\begin{figure}[tb]
		\centering
			\includegraphics[page=2]{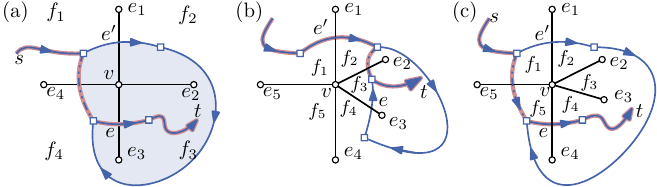}
		\caption{Inconsistent path around (a) a degree-4 vertex and (b,c) a degree-5
		vertex.}
		\label{fig:degree_45}
	\end{figure}
              
        Let $p$ be a shortest $st$-path in $\ExtG$.  We call an edge
        $e$ of $p$ \emph{good} if the vertices left and right of it do
        not have label $LR$ in the labeling $l_p$ induced by~$p$.

        If~$p$ is not consistent, then let $e$ denote the last edge
        of~$p$ that is not good.  Then an endpoint $v$ of the primal
        edge corresponding to $e$ has label $LR$.  Without loss of
        generality, we may assume that $v$ lies left of $e$.  Since
        $l_p(v) = LR$, there is an edge $e'$ of $p$ that has $v$ to
        its right.  By the choice of $e$, it follows that $e'$ lies
        before $e$ on $p$.
        We now distinguish cases based on the degree of~$v$.

        If~$\deg(v) \le 3$, then we find that $p$ enters or leaves a
        face twice, which contradicts the assumption that it is a
        shortest $st$-path.
        
        If~$\deg(v) = 4$, we denote the edges around $v$ in clockwise
        order as $e_1,\dots,e_4$ such that $e'$ crosses $e_1$.
        Moreover, we denote the faces incident to $v$ in clockwise
        order as $f_1,\dots,f_4$ where $f_1$ is the starting face of
        $e'$.

        Since no face has two incoming or two outgoing edges of $p$,
        it follows that $e' = f_1f_2$ crosses $e_1$ and $e = f_4f_3$
        crosses $e_3$; see Fig.~\ref{fig:degree_45}a.  Let $p'$ be the path obtained
        from $p$ by replacing the subpath~$p[f_1,f_4]$ by the edge
				$f_1f_4$ that crosses $e_4$.  Since $p$ is a shortest path, it follows
        that $f_2=f_4$.  By construction, it is~$l_{p'}(v) = L$.
        Observe that $p'[f_4,t] = p[f_4,t]$ lies inside the region
        $\rho$ bounded by $p[f_1,f_4]$ and a curve connecting $f_1$
        and $f_4$ that crosses $e_4$.  The only vertex inside this
        region whose label changed is $v$.  Therefore, the path
        $p'[f_1,t]$ consists of good edges, and we have thus increased
        the length of the suffix of the shortest path that consists of
        good edges.

        Now assume that $\deg(v) = 5$.  We denote the edges around $v$
        as $e_1,\dots,e_5$ in clockwise order such that $e'$ crosses
        $e_1$.  We further denote the faces incident to $v$ in
        clockwise order as $f_1,\dots,f_5$ such that $e'$ starts in
        $f_1$.  Since no face has two incoming or two outgoing edges,
        it follows that either $e$ crosses $e_4$ from $f_5$ to $f_4$
        or $e$ crosses $e_3$ from $f_4$ to $f_3$.

        If $e$ crosses $e_3$, then we consider the path $p'$ obtained
        from $p$ by replacing the subpath $p[f_2,f_3]$ by the 
        edge that crosses $e_3$; see Fig.~\ref{fig:degree_45}b.  As above, it
        follows that $f_2=f_4$ and $v$ is a cutvertex and that
        $p'[f_1,t]$ consists of good edges.

        If $e$ crosses $e_4$, then we obtain $p'$ by replacing
        $p[f_1,f_5]$ by the single edge that crosses $e_5$; see
        Fig.~\ref{fig:degree_45}c. As above, we find that $f_2 = f_5$ and $v$ is a
        cutvertex and that $p'[f_1,t]$ consists of good edges.

				Thus, in all cases, we increase the length of the suffix of the shortest
				path consisting of good edges.  Eventually, we thus arrive at a shortest
				path whose edges are all good and that hence is consistent.
\end{proof} 


\begin{theorem}
	Let $G=(V, E)$ a planar embedded graph with maximum vertex-degree $\Delta$ and
	let $p$ be a shortest $st$-path in $\ExtG$ with $s, t \in V$.
	Then there is a consistent path of length at most $(\Delta -2) |p|$.
\end{theorem}

\begin{proof}

	\begin{figure}[tb]
		\centering
		\includegraphics[page=5]{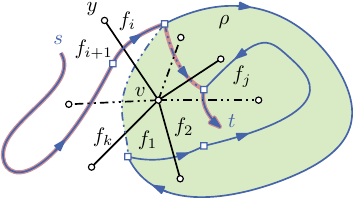}
		\caption{Inconsistent path around a degree $k$ vertex.}
		\label{fig:degree_k}
	\end{figure}

	Let $p$ be an $st$-path in $\ExtG$. Assume that $p$ is not consistent.  Then
	there is a shortest prefix $p_2 = p[s, f_2] = p[s, f_1] \cdot f_1f_2$ of  $p$
	that is not consistent; refer to Fig.~\ref{fig:degree_k}. Let $v$ be a
	vertex incident to the primal edge of $f_1f_2$ with $l_{p_2}(v) = LR$.
	Without loss of generality let $f_1, f_2, \dots, f_k$ be the faces around $v$
	in counterclockwise order, i.e., $v$ lies left of $f_1f_2$.  
	
	Since $p_2$ is not consistent, there is a second edge of $p_2$ that
	crosses a primal edge incident to $v$.  Let $e$ be the last edge of
	$p[s, f_1]$ that crosses a primal edge incident to $v$.  Since $p_2$ is the
	shortest inconsistent prefix of $p$, $v$ lies right of $e$, i.e., $e =
	f_{i+1}f_i$ for some $i$ with $ 2 < i \leq k - 1$.  Moreover, let $f_j$ be the first
	vertex in clockwise order from $f_i$ that lies on the path
	$p[f_2, t]$. Note that such a vertex $f_j$ exists, since at the latest $f_2$
	satisfies the condition.

	Let $q$ be the path $f_i f_{i-1} \cdots f_j$. We obtain a path $p'$ from $p$ by
	replacing $p[f_i, f_j]$ by $q$, i.e., $p' = p[s,f_i] \cdot q \cdot p[f_j,t]$.
	Note that, since $f_j$ is the first vertex in clockwise order on $p[f_2, t]$,
	$p'$ is a simple path.  Since $q$ does not contain the edges $f_kf_1$ and $f_1
	f_2$, and $p[f_i,f_j]$ contains at least one edge, the path $p'$ has length at
	most $|p| + (k - 2) - 1$.	We claim that the prefix $p'_j = p'[s,f_j]$ is
	consistent.  
	
	Then,  since $p'[f_j, t]$ is a subpath of $p[f_2, t]$ and 
	$p'[s,f_j]$ is consistent, it follows that we have decreased the maximum
	length of a suffix of the path whose removal results in an inconsistent path.
	Since this suffix has initially length at most $|p|$, we inductively find a
	consistent $st$-path of length at most $(\Delta - 2) |p|$.

	It remains to prove that $p'[s, f_j]$ is consistent.  Since $p[s, f_2]$ is the
	shortest inconsistent prefix of $p$, the prefix $p[s, f_1]$ is consistent.
	Therefore, $v$ is right of $p[s, f_i] = p'[s, f_i]$. By construction, $v$ is
	right of $q$. Thus, we have $l_{p'_j}(v) = R$.  The only vertices $w$ of $G$
	with $l_{p'_j}(w) = LR$ can be neighbors of $v$, as otherwise $p[s, f_1]$
	would not be consistent.

	Consider the region $\rho$ enclosed by the path $p[f_i, f_1]$ and $f_1, f_k,
	\dots, f_i$ that contains $v$; refer to Fig.~\ref{fig:degree_k}.  The prefix
	$p[s, f_1] = p'[s, f_1]$ lies outside of $\rho$ and the path $q$ lies entirely
	in $\rho$.  Moreover, in case that $vw$ is crossed by $p'[s, f_i]$, $w$ lies
	outside of $\rho$. On the other hand, if $q$ crosses an edge $vw$, then $w$
	lies inside $\rho$.  Thus, in both cases we immediately get that
	$l_{p'_j}(w) = L$. Therefore, the prefix $p'[s, f_j]$ is consistent.
\end{proof}

\section{Consistent Shortest $st$-paths}
\label{sec:consistent_shortest}

In Section~\ref{sec:bounded_degree} we showed that every shortest $st$-path in
the extended dual $\ExtG$ of a graph $G$ with vertex degree at most $3$ is
consistent. For every graph of maximum degree $5$, there is a shortest $st$-path
$\ExtG$ that is consistent. On the other hand, Fig.~\ref{fig:bad_ratio} shows
that, starting from degree $6$, there are graphs whose shortest $st$-paths are
not consistent.  In this section we investigate the problem of deciding whether
$\ExtG$ contains a consistent shortest $st$-path. As a consequence of
Proposition~\ref{prop:consistent_path} this problem is in $\cNP$.

In Lemma~\ref{lemma:edge_disjoint} we show that finding a consistent $st$-path
$p$ in $\ExtG$ is closely related to finding two edge-disjoint paths in $G$.
Especially, we are interested in two edge-disjoint paths where the length of one
is minimized. Eilam-Tzoreff~\cite{EILAMTZOREFF1998113} proved that this problem
is in general $\cNP$-complete.  In planar graphs the sum of the length of two
vertex-disjoint paths can be minimized efficiently~\cite{KOBAYASHI2010234}. In
general directed graphs the problem is $\cNP$-hard~\cite{FORTUNE1980111}.
Finding two edge-disjoint paths in acyclic directed graphs is
$\cNP$-complete~\cite{4567876}. 

The closest relative to our problem is certainly the work of Eilam-Tzoreff.
In fact their result can be modified to show that it is $\cNP$-hard to decide whether a
graph contains two edge-disjoint $st$-paths such that one of them is a shortest
path.  We study this problem in the planar setting with the additional
restriction that $s$ and $t$ lie on a common face of the subgraph $\Gsp$ of
$\ExtG$ that contains all shortest paths from $s$ to $t$.

\begin{lemma}
 \label{lemma:edge_disjoint}
 An $st$-path $p$ in $\ExtG$ is consistent if and only if there is an $st$-path
 $p'$ in $\ExtG$ that is edge-disjoint from $p$ and that does not cross $p$.
\end{lemma}

\begin{proof}
	\begin{figure}[tb]
		\centering
			\includegraphics[page=1]{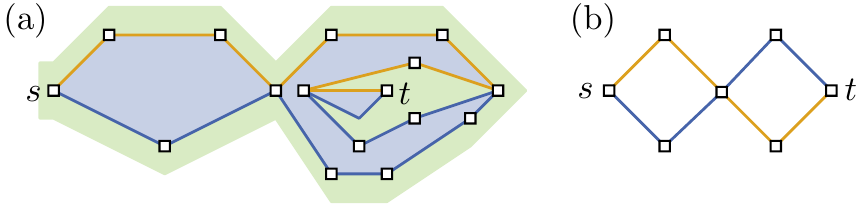}
		\caption{(a) The green regions are right of $p$ (blue) and the blue left of
		$p$. (b) The outer region that is not bounded by 
		maximal subpaths of $p$ and $p'$.}
		\label{fig:edge_disjoint_paths}
	\end{figure}

	The paths $p$ and $p'$ define a set of regions in the plane.  Since $p$ and
	$p'$ are non-crossing, each region is bounded by one maximal subpath of $p$
	and one maximal subpath of $p'$ (Fig.~\ref{fig:edge_disjoint_paths}). We label
	each region $\rho$ with either $L$ or $R$, depending on whether $\rho$ lies
	left or right of the unique maximal subpath of $p$ on its boundary. We define
	a labeling $l$ of $G$ by giving each vertex $v$ the label of the region $\rho$
	that contains it. We claim that $l$ is compatible with $p$.

	Since $p$ and $p'$ are edge-disjoint, every primal edge connects vertices
	of the same or adjacent regions. Moreover, by construction,
	vertices of adjacent regions have different labels. Thus all vertices left of
	$p$ have label $L$ and all vertices right of $p$ have label $R$. That is $l$
	is compatible with $p$, i.e., $p$ is consistent.

	\begin{figure}[tb]
		\centering
		\includegraphics[page=3]{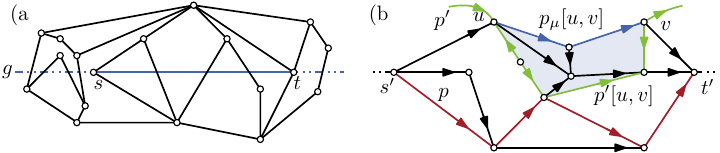}
		\caption{(a) The line $g$ through the segment $st$ induces a path in
		$\ExtG$. (b) Modification of the undirected path $p'$ edge-disjoint from $p$.}
		\label{fig:drawing_edge_disjoint}
	\end{figure}

	Conversely, assume that $p$ is consistent. By
	Proposition~\ref{prop:consistent_path} there is a straight-line drawing of $G$
	such that the segment $st$ intersects the same edges as $p$ and in the same
	order (Fig.~\ref{fig:drawing_edge_disjoint}a). Let $g$ be the line that
	contains the segment $st$. 
	%
	Each edge of $G$ intersects $g$ at most once.	
	Thus, the complement of $st$ in $g$ defines a path from $s$ to $t$ in
	$\ExtG$ that is edge-disjoint from $p$ and does not cross~$p$.
\end{proof} 

Thus, we now consider the problem of finding a consistent shortest
$st$-path as an edge-disjoint path problem in $\ExtG$.  Our proof
strategy consists of three steps.
\begin{inparaenum}[Step 1)]
\item We first show that the problem is equivalent
to finding two edge-disjoint paths $p$ and $q$ in a directed graph
$\DirG$ such that $p$ is directed and $q$ is undirected.
\item We modify
$\DirG$ such that $p$ is a path in a specific subgraph $\Gsp$ and $q$
lies in the subgraph $\coGsp$. These two graphs may share an edge set
$\hat{E}$ such that each edge in $\hat{E}$ can be an edge of
$p$ or of $q$. Moreover, we find pairs of edges $e$ and $e'$ in
$\hat{E}$ such that the path $p$ in $\Gsp$ (the path $q$ in $\coGsp$)
contains either $e$ or $e'$.
\item Finally, we use these properties to
reduce our problem to $2$-SAT.
\end{inparaenum}

We begin with Step 1.
A directed graph $\DirG=(V' \cup \{s, t\}, E')$ is \emph{$st$-friendly}
if $\ExtG$ contains a consistent shortest $st$-path if and only if $\DirG$
contains a directed $st$-path $p$ and an undirected $st$-path $p'$ that is
edge-disjoint from $p$ and does not cross $p$. We obtain an $st$-friendly graph
$\DirG = (\DirV, \DirE)$ from $\ExtG$ as follows.  Denote by $\Gsp$ the directed
acyclic graph that contains all shortest paths from $s$ to $t$ in $\ExtG=(V,
E)$.  If an edge $uv \in E$ is an edge of $\Gsp$, we add it to $\DirG$. For all
remaining edges $uv$, we add a subdivision vertex $x$ to $\DirG$ and add the
directed edges $xu, xv$ to $\DirG$ in this direction. We claim that $\DirG$ is
$st$-friendly.

Let $p$ be a consistent shortest $st$-path in $\ExtG$. By
Lemma~\ref{lemma:edge_disjoint} there is a path $p'$ in $\ExtG$ that is
edge-disjoint from $p$ and does not cross $p$. By construction $p$ corresponds
to a directed path in $\Gsp$ and $p$ corresponds to an undirected path in
$\DirG$. Conversely,  due to the directions of the edges $xv, xu$, every
directed $st$-path $q$ in $\DirG$ is a directed path in $\Gsp$, and therefore it
is a shortest $st$-path in $\ExtG$. If there is an undirected path $q'$ that is
edge-disjoint from $q$ and does not cross $q$, we obtain a path $p'$ from $q'$
by contracting edges incident to split vertices $x$. Hence, $\DirG$ is
$st$-friendly.
 
We consider the following special case, where $s$ and $t$ lie on a common face
$o$ of the subgraph $\Gsp$ of $\DirG$. Without loss of generality, let $o$ be
the outer face of $\Gsp$ and let $t$ lie on the outer face of $\DirG$.  We
denote by $\pu$ and  $\pl$ the upper and lower $st$-path of $\Gsp$ on the
boundary of $o$.
A vertex $v$ of $\Gsp$ is an \emph{interior vertex} if $v$ does not lie on
$o$. An edge $uv$ of $\Gsp$ is an \emph{interior edge} if $u$ and $v$ are
interior vertices.  An edge $e$ of $\Gsp$ is a \emph{chord} if both its
endpoints lie on $o$ but $e$ is not an edge on the boundary of $o$.

\begin{lemma} \label{lemma:long_on_boundary}
	%
	For a directed $st$-path $p$ and an undirected $st$-path $p'$, that are
	edge-disjoint and non-crossing, there is an undirected $st$-path $p''$ that is
	edge-disjoint from $p$, does not cross $p$, and that does not use interior
	vertices of $\Gsp$.
\end{lemma}

\begin{proof}
	%
  %
	Since $p$ and $p'$ are non-crossing, there are two distinct vertices $u,v$ on
	$\pl$ or on $\pu$, say $\pu$, such that the inner vertices of $p'[u,v]$ lie in
	the interior of $\Gsp$; refer to Fig.~\ref{fig:drawing_edge_disjoint}b. Moreover, since
	$p'$ and $p$ are non-crossing,  the region enclosed by $p'[u,v]$ and
	$\pu[u,v]$ does not contain a vertex of $p$ in its interior.  Therefore, we
	obtain $p''$ by iteratively replacing pieces in the form of $p'[u,v]$ by
	$\pu[u,v]$.
\end{proof}


\begin{figure}
	\centering
	\includegraphics{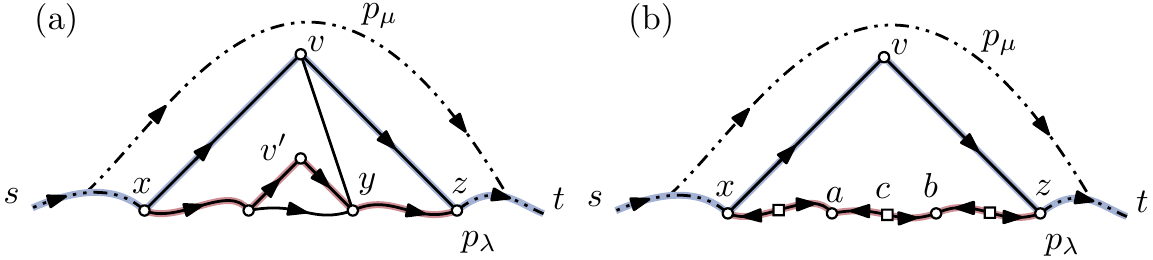}
	\caption{
		(a) The red directed path can be circumvented with the blue directed path via vertex $v$.
		(b) The red path consists of avoidable edges.	
	}
	\label{fig:interior_rerouting}
\end{figure}

This finishes Step 1, and we continue with Step 2.
In the following, we iteratively simplify the structure of $\Gsp$ while
preserving $st$-friendliness of $\DirG$.  Due to
Lemma~\ref{lemma:long_on_boundary}, the graph $\Gsp/e$, obtained from
contracting an edge $e$ of $\Gsp$, is $st$-friendly, if $e$ is an interior edge.
This may generate a separating triangle $xyz$. Let $v$ be a vertex in the
interior of $xyz$ and let $p$ be a directed $st$-path that contains $v$. Then,
$p$ contains at least two vertices of $x,y,z$. Hence, $p$ can be rerouted using
an edge of $xyz$. Thus, the graph after removing all vertices in the interior of
$xyz$ is $st$-friendly.  After contracting all interior edges of $\Gsp$, each
neighbor of an interior vertex of $\Gsp$ lies either on $\pl$ or  on $\pu$.  The
remaining edges are edges on $\pl \cup \pu$ and chords.

Consider three vertices $x,y,z$ that lie in this order on $\pl$ ($\pu)$ and two
interior vertices $v$ and $v'$, with $xv, v'y, vz \in \DirE$; refer to
Fig.~\ref{fig:interior_rerouting}a. Note that $v$ and $v'$ can coincide.  Then,
every directed $st$-path $p$ that contains $y$ also contains $x$ and $z$.
Hence, $p$ can be rerouted through the edges $xv, vz$ and as a consequence of
Lemma~\ref{lemma:long_on_boundary}, the graph $\Gsp-v'y$ is $st$-friendly.
Analogously, if $\Gsp$ contains the edge $yv'$, $\Gsp-yv'$ remains
$st$-friendly.  We call such edges \emph{circumventable}.

We refer to edges of a subpath $\pl[x, z]$ ($\pu[x,z]$) as \emph{avoidable} if
there exists an interior vertex $v$ with $xv, vz \in \DirE$
(Fig.~\ref{fig:interior_rerouting}b).  If there exists a directed path $p$ that
uses an avoidable edge $ab$ it can be rerouted by replacing the corresponding
path $\pl[x,z]$ with the edges $xv,vz$.  Thus, we can split the edge $ab$ with a
vertex~$c$ and we direct the resulting edges from $c$ towards $a$ and $b$,
respectively, and remove the edge $ab$ from $\DirG$.  Finally, we iteratively
contract edges incident to vertices with in- and out-degree $1$, and we
iteratively remove vertices of degree at most~$1$, except for $s$ and $t$.
\begin{figure}[tb]
	\centering
	\includegraphics[page=3]{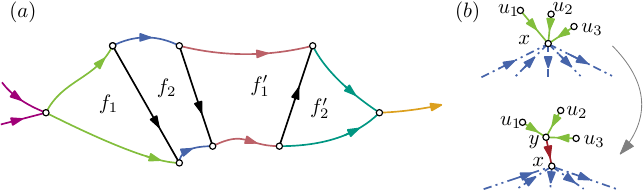}
	\caption{(a) Interior partners decoded by color of $2$-edge connected component of
		$\Gsp$. (b) Split a vertex $x$ on the boundary of $\ExtH$.}
	\label{fig:interior_partners}
\end{figure}
Since all interior edges of $\Gsp$ are contracted, circumventable interior edges
are removed and avoidable edges are replaced, each $2$-edge connected component
of $\Gsp$ is an outerplanar graph whose weak dual (excluding the outer face) is
a path; compare Fig.~\ref{fig:interior_partners}a. Each face $f$ of $\Gsp$, with
$f\not=o$, contains at least one edge $e_\lambda$ of $\pl$ and one edge $e_\mu$
on $\pu$.  Moreover, every directed $st$-path contains either $e_\lambda$ or $e_\mu$.
We refer to the edge sets $E_{f,\lambda} = E(f) \cap E(\pl)$ and $E_{f, \mu} =
E(f) \cap E(\pu)$ as \emph{interior partners}.

\begin{property}
	\label{prop:directed_path}
	Choosing a directed $st$-path in $\Gsp$ is equivalent to choosing for each
	face $f$ of $\Gsp$ one of the interior partners $E_{f,\mu}$ or $E_{f,\lambda}$
	such that the following condition holds.  Let $f_1, f_2$ be two adjacent faces
	that are separated by a chord $e$ that ends at $\pl$ ($\pu$) such that $f_1$
	is right of $e$ (left of $e$), then the choice of $E_{f_2, \mu}$ ($E_{f_2,
	\lambda}$) implies the choice of $E_{f_1, \mu}$ ($E_{f_1, \lambda}$).
\end{property}

In the following, we modify the \emph{exterior of $\DirG$}, i.e., $\coGsp =
\DirG - E(\Gsp)$, with the aim to obtain an analog property for the choice of
the undirected path.  We refer to edges of $\coGsp$ as \emph{exterior edges}. A
vertex  in $V(\coGsp) \setminus V(\Gsp)$ is an \emph{exterior vertex}.  

Since the undirected path is not allowed to cross the directed path, we split
each cut vertex $x$ into an upper copy $x_\mu$ and a lower copy $x_\lambda$. We
reconnect edges of $\pl$ and $\pu$ incident to $x$ to $x_\lambda$ and $x_\mu$,
respectively. Exterior edges incident to $x$ that are embedded to the right of
$\pl$ are reconnected to $x_\lambda$. Likewise, edges embedded to the left of
$\pu$ are reconnected to $x_\mu$. Note that this operation duplicates  bridges
of $\Gsp$. Thus, we forbid the undirected path to traverse these duplicates.
Observe that after this operation the outer face $o$ of $\Gsp$ is bounded by a
simple cycle.

Let $x$ be a vertex on $o$ that is incident to an exterior edge. In this case,
we insert a vertex $y$ to $\DirG$ and we remove each exterior edge $ux$ from
$\DirG$ and insert as a replacement edges $yx$ and $yu$; see
Fig.~\ref{fig:interior_partners}b. We refer to the edge $yx$ as a \emph{barrier}.
Since the barrier $yx$ is directed from $y$ to $x$, the modification
preserves the $st$-friendliness of $\DirG$.  We now exhaustively contract
exterior edges that are not barrier edges, and remove vertices in the interior
of separating triangles.

\begin{figure}[tb]
	\centering
	\includegraphics{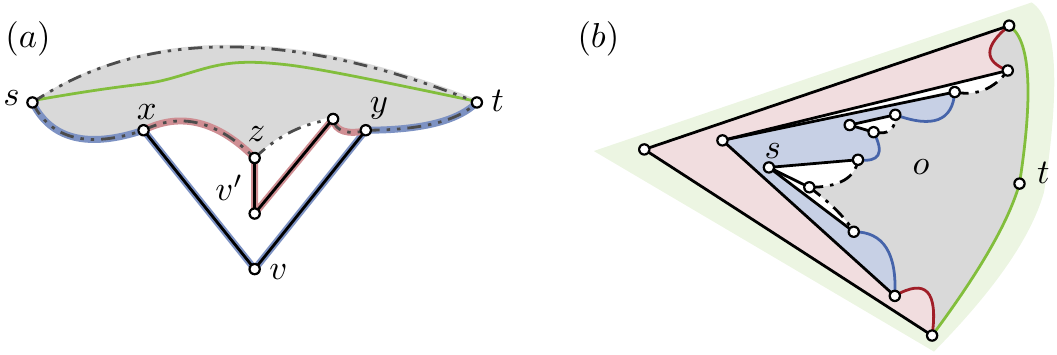}
	\caption{(a) If the undirected path contains $z$, it can be rerouted to use vertex
	$v$. (b) The color coding of the faces indicate the exterior partners.}
	\label{fig:exterior}
\end{figure}

Recall that $s$ and $t$ lie on a common face $o$ of the subgraph $\Gsp$ of
$\DirG$ and $t$ lies on the outer face of $\Gsp$.  Let $v$ be an exterior vertex
such that its neighbor $x$ comes before its neighbor $y$ on $p_i, i = \lambda,
\mu$, refer to Fig.~\ref{fig:exterior}a. Let $z$ be a vertex between $x$ and $y$
on $p_i$ that is connected to a vertex $v'$ such that the edge $v'z$ ($zv'$) lies in
the interior of the region bounded by $yvx$ and $p_i[x,y]$.  Consider a directed
$st$-path $p$ in $\Gsp$ and an undirected $st$-path $p'$ in $\DirG$ that is
edge-disjoint from $p$, that does not cross $p$ and that contains $v'$. Due to
Lemma~\ref{lemma:long_on_boundary} we can assume, that $p'$ does not contain an
interior vertex of $\Gsp$. Thus, it contains $x$ and $y$. We obtain a new path
$p''$ by replacing the subpath $p'[x,y]$ by $vx, vy$. Since $vx, vy$ are
exterior edges, $p''$ and $p$ are edge-disjoint and non-crossing. Thus, the
graph $\DirG-v'z$ ($\DirG - zv'$) is $st$-friendly.  After removing all such
edges, for any two neighbors $x$ and $y$ of an exterior vertex $v$, the paths
$o[x,y]$ and $o[y,x]$ each contains either $s$ and $t$.  Hence, the region bounded by
$yvx$ and $o[x, y]$ contains a second exterior vertex $v'$ if and only if
$o[x,y]$ contains either $s$ or $t$.  

Hence, the dual of $\coGsp$, with the dual vertex of $o$ removed, is a
caterpillar $C$, refer to Fig.~\ref{fig:exterior}b.  In case that $s$ or $t$ is
incident to an exterior vertex $v$, we can assume that the undirected path $p'$
contains the edge $sv$ ($vt$). Thus, for simplicity, we now assume that neither
$s$ nor $t$ is connected to an exterior vertex.  Let $a$ and $b$ be the vertices
in $C$ whose primal faces are incident to $s$ and $t$, respectively.  Then every
undirected $st$-path in $\coGsp$ from $s$ to $t$ traverses the primal faces of
the simple path $q$ from $a$ to $b$ in $C$.  Let $f$ be a primal face of a
vertex on $q$. Since we inserted the barrier edges to $\DirG$, every face
contains at least one edge $e_\lambda$ of $\pl$ and one edge $e_\mu$ of $\pu$.
Therefore, every undirected $st$-path in $\coGsp$ either contains $e_\lambda$ or
$e_\mu$. We refer to the sets $E_{f,\lambda} = E(f) \cap E(\pl)$ and
$E_{f,\mu} = E(f) \cap E(\pu)$ as \emph{exterior partners}.

\begin{property}
	\label{prop:undireted_path}
	Choosing an undirected $st$-path in $\coGsp$ is equivalent to choosing for
	each face $f \not = o$ of $\coGsp$ one of the exterior partners
	$E_{f,\lambda}$ or $E_{f, \mu}$.
\end{property}

This finishes Step 2, and we proceed to Step 3.  The problem of
finding a directed $st$-path $p$ and an undirected $st$-path $p'$ in
$\DirG$ reduces to a $2$-SAT instance as follows.  For each exterior
and interior partner we introduce variables $x_f$ and $x_g$,
respectively, where $f$ and $g$ correspond to the faces of the
partners.  If $x_f$ is true, $p'$ contains the edge of
$E_{f, \lambda}$, otherwise it contains $E_{f,\mu}$.  The conditions
on the choice of $p$ in Property~\ref{prop:directed_path} can be
formulated as implications.  Let $E_{f, \mu}$ an $E_{f, \lambda}$ be
exterior partners and let $E_{g, \mu}$ and $E_{g,\lambda}$ be interior
partners. In case that
$E_{f,\lambda} \cap E_{g,\lambda} \not=\emptyset$, either $p$ can
contain edges of $E_{g,\lambda}$ or $p'$ can contains edges of
$E_{f,\lambda}$ but not both.  Thus, $x_f$ and $x_g$ are not allowed
to be true at the same time, i.e., $x_f = \overline{x_g}$. Hence, we
have the following Theorem.

\begin{theorem}
	If $s$ and $t$ lie on a common face of $\Gsp$, it is decidable in polynomial time
	whether $\DirG$ has a directed $st$-path and an undirected $st$-path that are
	edge-disjoint and non-crossing.
\end{theorem} 

\begin{corollary}
	If $s$ and $t$ lie on a common face of $\Gsp$, it is decidable in polynomial
	time whether $\ExtG$ contains a consistent shortest $st$-path.
\end{corollary}

\section{Parametrized Complexity of Short Consistent $st$-Paths}
\label{sec:fpt} 
In this section we show that edge insertion can be solved in FPT time
with respect to the minimum number of crossings of a straight-line
drawing of $G+st$ where $G$ is drawn without crossings and has the
specified embedding.  Let~$l$ be an arbitrary labeling of~$G$.  Observe
that~$l$ defines a directed subgraph of~$\ExtG$ by removing each edge
whose dual edge has endpoints with the same label and by directing all other edges $e$
such that the endpoint of its primal edge left of $e$ has label~$L$
and its other endpoint has label~$R$.  We denote this graph
by~$\ExtG(l)$.  Obviously, a shortest $st$-path in~$\ExtG(l)$ is
compatible with~$l$, and thus a corresponding drawing exists.
Clearly, given the labeling~$l$ a shortest $st$-path in~$\ExtG(l)$ can
be computed in linear time by a BFS.

Now assume that the length of a shortest consistent path in~$\ExtG$
is~$k$.  We propose a randomized FPT algorithm with running
time~$O(4^k n)$ for finding a shortest consistent path in~$\ExtG$, based on the
color-coding technique~\cite{DBLP:books/sp/CyganFKLMPPS15}. 

The algorithm works as follows.  First, we pick a random labeling
of~$G$ by labeling each vertex independently with~$L$ or $R$ with
probability~$1/2$.  We then compute a shortest path
in~$\ExtG(l)$.  We repeat this process $4^k$ times and report the
shortest path found in all iterations.  

Clearly the running time is~$O(4^k n)$.  Moreover, each reported path is
consistent, and therefore the algorithm outputs only consistent paths.  It
remains to show that the algorithm finds a path of length~$k$ with constant
probability.

Consider a single iteration of the procedure.  If the random labeling~$l$ is
compatible with~$p$, then the algorithm finds a path of length~$k$.  Therefore
the probability that our algorithm finds a consistent path of length~$k$ is at
least as high as the probability that $p$ is compatible with the random labeling
$l$.  Let~$V_L,V_R \subseteq V$ denote the vertices of~$V$ that are left and
right of~$p$, respectively.  Clearly it is~$|V_L|,|V_R| \le k$.  A random
labeling~$l$ is consistent with $p$ if it labels all vertices in~$V_L$ with $L$
and all vertices in $V_R$ with $R$.  Since vertices are labeled independently
with probability $1/2$, it follows that~$\Pr[ p$ is consistent with $l] =
(1/2)^{|V_L|} \cdot (1/2)^{|V_R|} \ge (1/2)^{2k} = (1/4)^k$.

Therefore, the probability that no path of length~$k$ is found in
$4^k$ iterations is at most~$(1-(1/4)^{k})^{4^k}$, which is
monotonically increasing and tends to~$1/e \approx 0.368$.  Thus the
algorithm succeeds with a probability of~$1-1/e \approx 0.632$.  The
success probability can be increased arbitrarily to~$1-\delta$,
$\delta > 0$ by repeating the algorithm~$\log(1/\delta)$ times.  The
probability that each iteration fails is then bounded from above
by~$(1/e)^{\log 1/\delta} = 1/e^{\log{1/\delta}} = \delta$.  E.g., to
reach a success probability of~$99\%$, it suffices to do
$\log 100 \le 5$ repetitions.  The algorithm can be derandomized with
standard techniques~\cite{DBLP:books/sp/CyganFKLMPPS15}.

\begin{theorem}
  There is a randomized algorithm $\mathcal A$ that computes a consistent path of
  length~$k$ if one exists with a success probability of~$1-\delta$. The running
  time of $\mathcal A$ is $O(\log(\delta^{-1}) 4^k n)$.
\end{theorem}

\section{Conclusion}

We have shown that the problem of finding a short consistent $st$-paths in $\ExtG$
is tractable in special cases and fixed-parameter tractable in general. Whether
$\ExtG$ has a short consistent $st$-path is equivalent to the question of whether
$\ExtG$ has two edge-disjoint and non-crossing $st$-paths, where the length of
one path is minimized. Surprisingly, this is related to yet another purely graph
theoretic problem: does a directed graph $G$ have two edge-disjoint paths where
one is directed and the other is only undirected? By the result of
Eilam-Tzoreff~\cite{EILAMTZOREFF1998113} the former problem is in general
$\cNP$-hard. For planar graphs the computational complexity of these problems
remains an intriguing open question.

In this paper, we only considered planar graphs with a fixed combinatorial
embedding. Allowing for arbitrary embeddings opens new perspectives on the
problem and is interesting future work.

\bibliography{strings,references}

\newcommand{\bibalenex}[2]{Proceedings of the #1 Workshop on Algorithm
  Engineering and Experiments (ALENEX'#2)} \newcommand{\bibdac}[2]{Proceedings
  of the #1 Annual Design Automation Conference (DAC'#2)}
  \newcommand{\bibinvisau}[1]{Proceedings of the Australian Symposium on
  Information Visualisation (invis.au #1)}
  \newcommand{\bibieeepdp}[2]{Proceedings of the #1 IEEE Symposium on Parallel
  and Distributed Processing #2} \newcommand{\bibieeecs}[1]{Proceedings of the
  IEEE International Symposium on Circuits and Systems #1}
  \newcommand{\bibcccg}[2]{Proceedings of the #1 Canadian Conference on
  Computational Geometry (CCCG'#2)} \newcommand{\bibswat}[2]{Proceedings of the
  #1 Scandinavian Workshop on Algorithm Theory (SWAT'#2)}
  \newcommand{\bibipco}[2]{Proceedings of the #1 International Conference on
  Integer Programming and Combinatorial Optimization (IPCO'#2)}
  \newcommand{\bibsofsem}[2]{Proceedings of the #1 Conference on Current Trends
  in Theory and Practice of Computer Science (SOFSEM'#2)}
  \newcommand{\bibstoc}[2]{Proceedings of the #1 Annual ACM Symposium on Theory
  of Computing (STOC'#2)} \newcommand{\bibfocs}[2]{Proceedings of the #1 Annual
  Symposium on Foundations of Computer Science (FOCS'#2)}
  \newcommand{\bibsoda}[2]{Proceedings of the #1 Annual ACM-SIAM Symposium on
  Discrete Algorithms (SODA'#2)} \newcommand{\bibgd}[2]{Proceedings of the #1
  International Symposium on Graph Drawing and Network Visualization (GD'#2)}
  \newcommand{\bibinfovis}[1]{Proceedings of the IEEE Symposium on Information
  Visualization (InfoVis'#1)} \newcommand{\bibvis}[1]{Proceedings of the IEEE
  Conference on Visualization (Vis'#1)} \newcommand{\bibpvis}[1]{Proceedings of
  the IEEE Pacific Visualisation Symposium (PacificVis'#1)}
  \newcommand{\bibsoftvis}[2]{Proceedings of the #1 ACM Symposium on Software
  Visualization (SoftVis'#2)} \newcommand{\bibeurocg}[2]{Proceedings of the #1
  European Workshop on Computational Geometry (EuroCG'#2)}
  \newcommand{\bibsocg}[2]{Proceedings of the #1 Annual Symposium on
  Computational Geometry (SoCG'#2)} \newcommand{\bibwads}[2]{Proceedings of the
  #1 International Symposium on Algorithms and Data Structures (WADS'#2)}
  \newcommand{\bibwg}[2]{Proceedings of the #1 Workshop on Graph-Theoretic
  Concepts in Computer Science (WG'#2)} \newcommand{\bibgta}{Proceedings of the
  Conference at Graph Theory and Applications}
  \newcommand{\bibisaac}[2]{Proceedings of the #1 International Symposium on
  Algorithms and Computation (ISAAC'#2)} \newcommand{\bibcocoon}[2]{Proceedings
  of the #1 Annual International Conference on Computing and Combinatorics
  (COCOON'#2)} \newcommand{\bibtamc}[2]{Proceedings of the #1 Annual Conference
  on Theory and Applications of Models of Computation (TAMC'#2)}
  \newcommand{\bibicalp}[2]{Proceedings of the #1 International Colloquium on
  Automata, Languages and Programming (ICALP'#2)}
  \newcommand{\biblatin}[2]{Proceedings of the #1 Latin American Symposium
  (LATIN'#2)} \newcommand{\bibesa}[2]{Proceedings of the #1 Annual European
  Symposium on Algorithms (ESA'#2)}
\begin{thebibliography}{10}
\providecommand{\url}[1]{\texttt{#1}}
\providecommand{\urlprefix}{URL }
\providecommand{\doi}[1]{https://doi.org/#1}

\bibitem{DBLP:conf/soda/ChimaniGMW09}
Chimani, M., Gutwenger, C., Mutzel, P., Wolf, C.: {Inserting a Vertex into a
  Planar Graph}. In: \bibsoda{20th}{09}. pp. 375--383 (2009)

\bibitem{chimani_et_al:LIPIcs:2016:5922}
Chimani, M., Hlinen{\'y}, P.: {Inserting Multiple Edges into a Planar Graph}.
  In: Fekete, S., Lubiw, A. (eds.) \bibsocg{32nd}{16}. Leibniz International
  Proceedings in Informatics (LIPIcs), vol.~51, pp. 30:1--30:15. Schloss
  Dagstuhl--Leibniz-Zentrum fuer Informatik (2016).
  \doi{10.4230/LIPIcs.SoCG.2016.30}

\bibitem{DBLP:books/sp/CyganFKLMPPS15}
Cygan, M., Fomin, F.V., Kowalik, L., Lokshtanov, D., Marx, D., Pilipczuk, M.,
  Pilipczuk, M., Saurabh, S.: {Parameterized Algorithms}. Springer-Verlag
  Springer International Publishing (2015). \doi{10.1007/978-3-319-21275-3}

\bibitem{10.1007/978-3-319-21840-3_25}
Eades, P., Hong, S.H., Liotta, G., Katoh, N., Poon, S.H.: {Straight-Line
  Drawability of a Planar Graph Plus an Edge}. In: Dehne, F., Sack, J.R.,
  Stege, U. (eds.) \bibwads{14th}{15}. pp. 301--313 (2015).
  \doi{0.1007/978-3-319-21840-3\_25}

\bibitem{EILAMTZOREFF1998113}
Eilam-Tzoreff, T.: {The Disjoint Shortest Paths Problem}. Discrete Applied
  Mathematics  \textbf{85}(2),  113 -- 138 (1998).
  \doi{10.1016/S0166-218X(97)00121-2}

\bibitem{4567876}
Even, S., Itai, A., Shamir, A.: {On the complexity of time table and
  multi-commodity flow problems}. In: 16th Annual Symposium on Foundations of
  Computer Science (SFCS 1975). pp. 184--193 (Oct 1975).
  \doi{10.1109/SFCS.1975.21}

\bibitem{FORTUNE1980111}
Fortune, S., Hopcroft, J., Wyllie, J.: The directed subgraph homeomorphism
  problem. Theory of Computing Systems  \textbf{10}(2),  111 -- 121 (1980).
  \doi{https://doi.org/10.1016/0304-3975(80)90009-2}

\bibitem{garey1983crossing}
Garey, M.R., Johnson, D.S.: {Crossing Number is NP-Complete}. SIAM Journal on
  Algebraic and Discrete Methods  \textbf{4}(3),  312--316 (1983)

\bibitem{JGAA-160}
{Gutwenger}, C., {Klein}, K., {Mutzel}, P.: {Planarity Testing and Optimal Edge
  Insertion with Embedding Constraints}. Journal of Graph Algorithms and
  Applications  \textbf{12}(1),  73--95 (2008). \doi{10.7155/jgaa.00160}

\bibitem{DBLP:journals/algorithmica/GutwengerMW05}
Gutwenger, C., Mutzel, P., Weiskircher, R.: {Inserting an Edge into a Planar
  Graph}. Algorithmica  \textbf{41}(4),  289--308 (2005).
  \doi{10.1007/s00453-004-1128-8}

\bibitem{KOBAYASHI2010234}
Kobayashi, Y., Sommer, C.: {On Shortest Disjoint Paths in Planar Graphs}.
  Discrete Optimization  \textbf{7}(4),  234 -- 245 (2010).
  \doi{https://doi.org/10.1016/j.disopt.2010.05.002}

\bibitem{DBLP:journals/corr/abs-1201-3011}
Kobourov, S.G.: {Force-Directed Drawing Algorithms}. In: Tamassia, R. (ed.)
  {Handbook of Graph Drawing and Visualization}, pp. 383--408. Chapman and
  Hall/CRC (2013)

\bibitem{doi:10.1137/1.9781611975055.12}
Radermacher, M., Reichard, K., Rutter, I., Wagner, D.: {A Geometric Heuristic
  for Rectilinear Crossing Minimization}. In: Pagh, R., Venkatasubramanian, S.
  (eds.) \bibalenex{20th}{18}. pp. 129--138 (2018).
  \doi{10.1137/1.9781611975055.12}

\end{thebibliography}

\appendix
\newpage

\section{Proof of Theorem~\ref{theorem:degree_3}}

\thmdegthree*

\begin{proof}
	Let $p$ be a shortest path in $\ExtG$. Assume that $p$ is not consistent. Then
	there is a  vertex $v$  that left and right of $p$. Let $fg$ be the first edge
	of $p$ that crosses a primal edge incident to $v$.  If the degree of $v$ is at
	most $2$, 	 then $p$ contains either a loop or a double edge, contradicting
	the assumption that $p$ is a shortest path.  Therefore, assume that the degree
	of $v$ is $3$.  Without loss of generality, let $f, g$ and $h$ be the faces
	around $v$ in clockwise order (Fig.~\ref{fig:degree_3}b).  Since $v$ is left
	and right of $p$, $p$ contains either the edge $fh$ or $hg$. Thus, $p$
	contains either $f$ or $g$ twice. This contradicts the assumption that $p$ is
	a shortest path. 
\end{proof}

\begin{figure}[tb]
	\centering
	\includegraphics{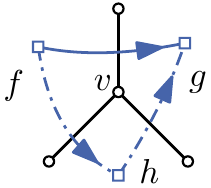}
	\caption{Inconsistent path around a
	degree-3 vertex.}
	\label{fig:degree_3}
\end{figure}

\end{document}